\documentclass[a4paper,12pt]{article}
\usepackage[a4paper, margin = 1.8cm]{geometry}
\usepackage[margin=1cm]{caption}

\usepackage{graphicx}%
\usepackage{multirow}%
\usepackage{amsmath,amssymb,amsfonts}%
\usepackage{amsthm}%
\usepackage{mathrsfs}%
\usepackage[title]{appendix}%
\usepackage{xcolor}%
\usepackage{textcomp}%
\usepackage{manyfoot}%
\usepackage{booktabs}%
\usepackage{algorithm}%
\usepackage{algorithmicx}%
\usepackage[noend]{algpseudocode}%
\usepackage{listings}%

\usepackage{hyperref}
\usepackage{cleveref}
\usepackage{subcaption}
\usepackage{placeins}

\DeclareMathOperator{\CC}{CC}
\newcommand{\size}[1]{\left\lvert#1\right\rvert}

\newtheorem{theorem}{Theorem}

\newtheorem{lemma}{Lemma}

\theoremstyle{thmstyletwo}%

\theoremstyle{thmstylethree}%
\newtheorem{definition}{Definition}%

\begin{document}

\title{Massively-Parallel Heat Map Sorting and\\ Applications To Explainable Clustering}

\author{Sepideh Aghamolaei and Mohammad Ghodsi \thanks{Department of Computer Engineering, Sharif University of Technology.\\sepideh.aghamolaei14@sharif.edu, ghodsi@sharif.edu}}

\maketitle

\abstract{Given a set of points labeled with $k$ labels, we introduce the heat map sorting problem as reordering and merging the points and dimensions while preserving the clusters (labels). A cluster is preserved if it remains connected, i.e., if it is not split into several clusters and no two clusters are merged.

We prove the problem is NP-hard and we give a fixed-parameter algorithm with a constant number of rounds in the massively parallel computation model, where each machine has a sublinear memory and the total memory of the machines is linear.
We give an approximation algorithm for a NP-hard special case of the problem.
We empirically compare our algorithm with k-means and density-based clustering (DBSCAN) using a dimensionality reduction via locality-sensitive hashing on several directed and undirected graphs of email and computer networks.}
\section{Introduction}

Simply put, an explainable model is any model that can be interpreted by a human by just looking at the output. For a review of explainable machine learning, see~\cite{marcinkevivcs2020interpretability}. One way to design explainable models is to minimize the complexity of the visual representations of clustering models/outputs while keeping the cost of clustering the same.

In the case of clustering, the definition of explainable clustering refers to a dimensionality reduction based on a given clustering.
This is based on a method of explainable clustering for $k$-medians and $k$-means called iterative mistake minimization~\cite{dasgupta2020explainable}. It computes a clustering using any algorithm and then it builds a threshold tree to separate the centers based on the number of misclassifications for each single-dimension (feature) split.

However, the classic way of defining the number of outliers as input does not work in the case of explainable clustering since for example, removing all the points of a small cluster would adversely affect the explainability. So, a different approach is needed.
Another way of improving the explainability of clustering is to use a tree of depth $k$ instead of a tree with $k$ leaves, which was empirically studied~\cite{laber2023shallow}. This is a special case of increasing the number of dimensions in the explainable model.

The number of clusters in the explainable clustering $(k)$ is not necessarily equal to the number of clusters in the original $k$-medians or $k$-means, for example, if two dimensions (features) together separate two clusters and none of them alone can separate them.
Setting the right value for $k$ requires executing the algorithm several times (at least $\min(k,\log n)$ times, if linear search and binary search are used). The resulting time complexity is too high to be used on big data.

A more general model than decision trees is needed for non-convex clusters (such as density-based clustering).
So, instead of using a tree, we use a simplified version of heat maps.
Heat maps represent data in the form of matrices with cells colored based on their values that are often used to represent the correlation between two sets of values. The dimensions are usually values of a categorical/nominal attribute. In a clustering heat map, the values of the dimensions are arranged based on a dendrogram (a tree representation of a hierarchical clustering built by merging the closest pairs, starting from individual items).

See \Cref{fig:heatmap} for an example.
\begin{figure}[h]
\centering
\begin{subfigure}[t]{0.45\textwidth}
\centering
\includegraphics[scale=0.5]{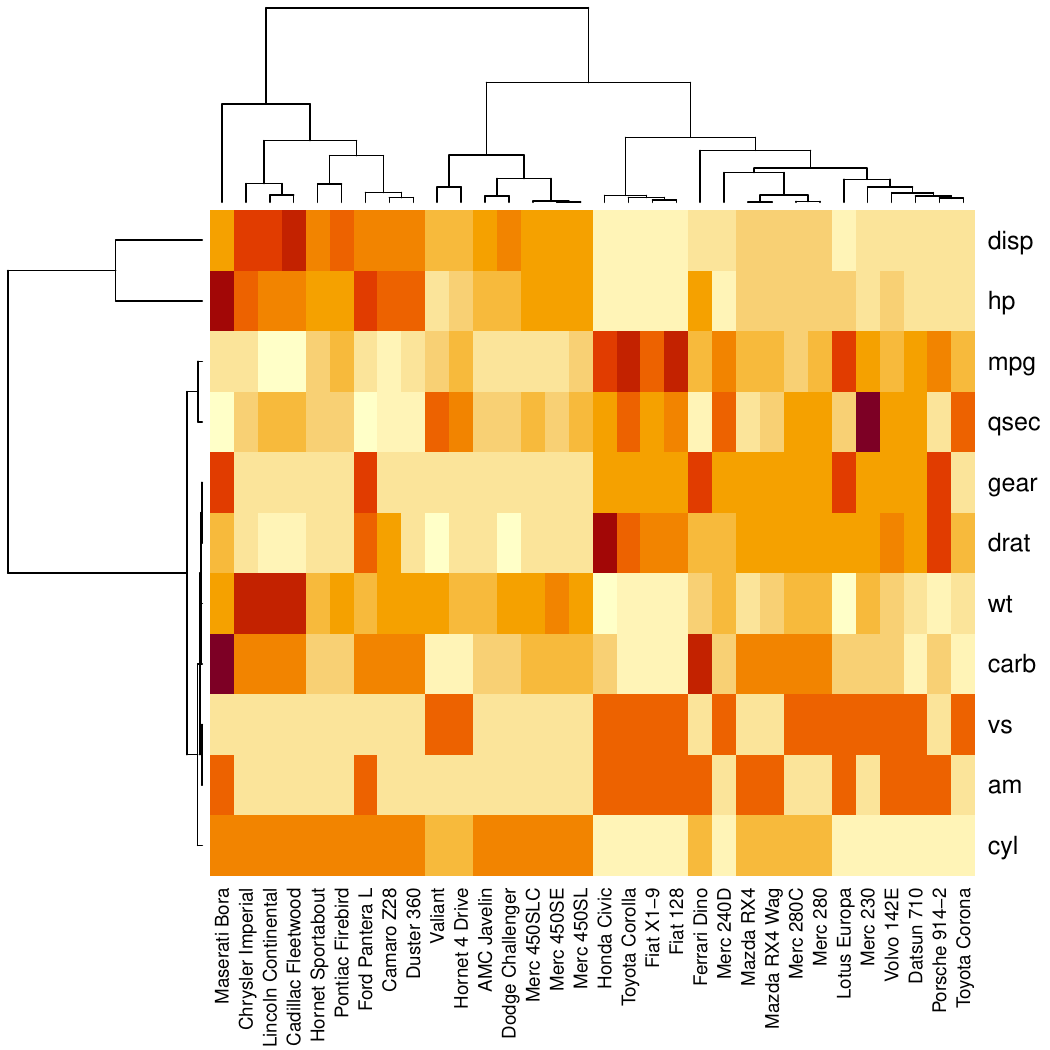}
\caption{Sorting based on dendrograms}\label{subfig:a}
\end{subfigure}
\begin{subfigure}[t]{0.45\textwidth}
\centering
\includegraphics[scale=0.5]{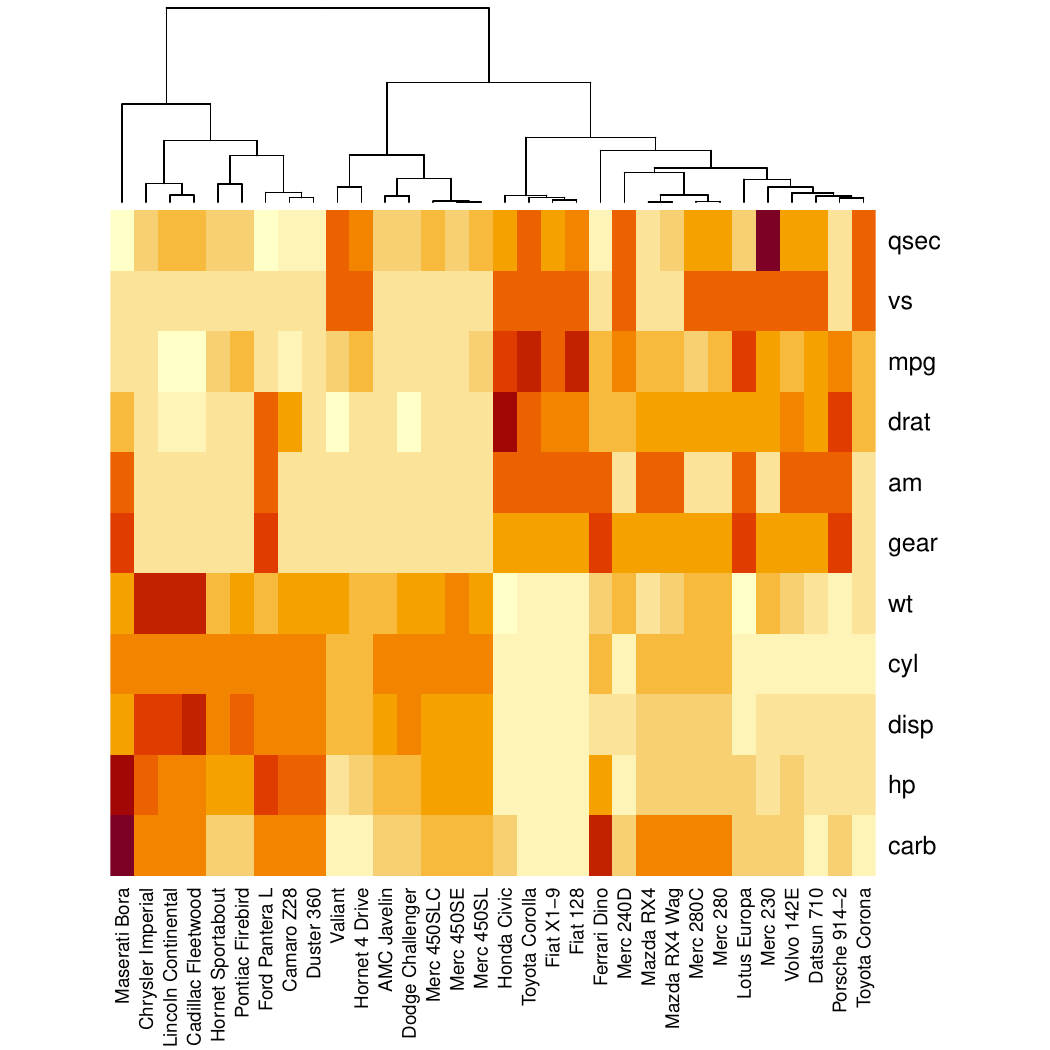}
\caption{Sorting the attributes (rows)}\label{subfig:b}
\end{subfigure}\\
\qquad\qquad
\begin{subfigure}[t]{0.4\textwidth}
\centering
\includegraphics[scale=0.4]{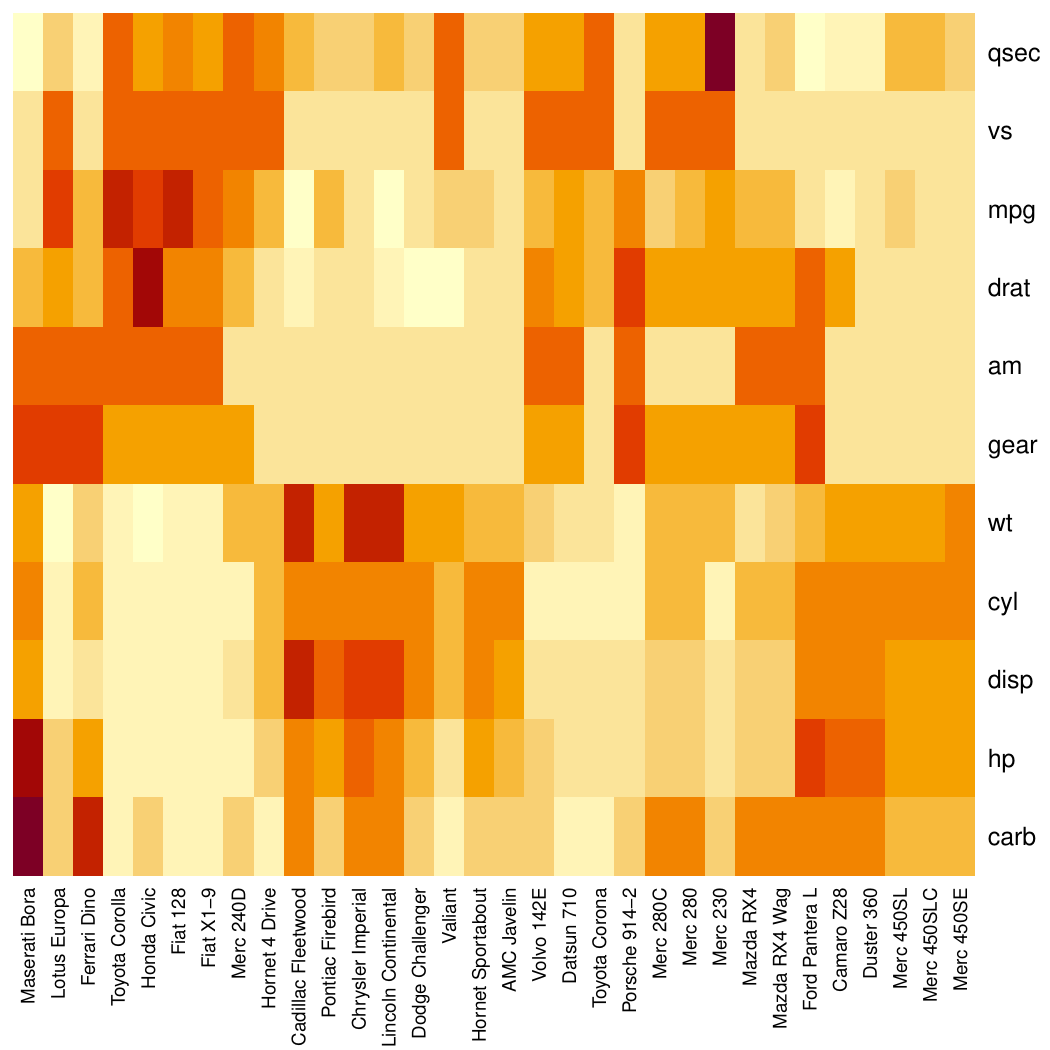}
\caption{Sorting based on correlations}\label{subfig:e}
\end{subfigure}
\caption{A heat map with dimension values sorted for mtcars dataset in R. First, an unsorted heatmap (\Cref{subfig:a}) is sorted by applying the ordering of the correlation matrix of the attributes (\Cref{subfig:c}) on the rows, resulting in the heat map of (\Cref{subfig:b}). However, applying the correlation matrix of points (\Cref{subfig:d}) to sort the columns creates an unsorted heat map on the clusters (\Cref{subfig:e}).}
\label{fig:heatmap}
\end{figure}

\begin{figure}
\centering
\begin{subfigure}[t]{0.9\textwidth}
\centering
\includegraphics[scale=0.7]{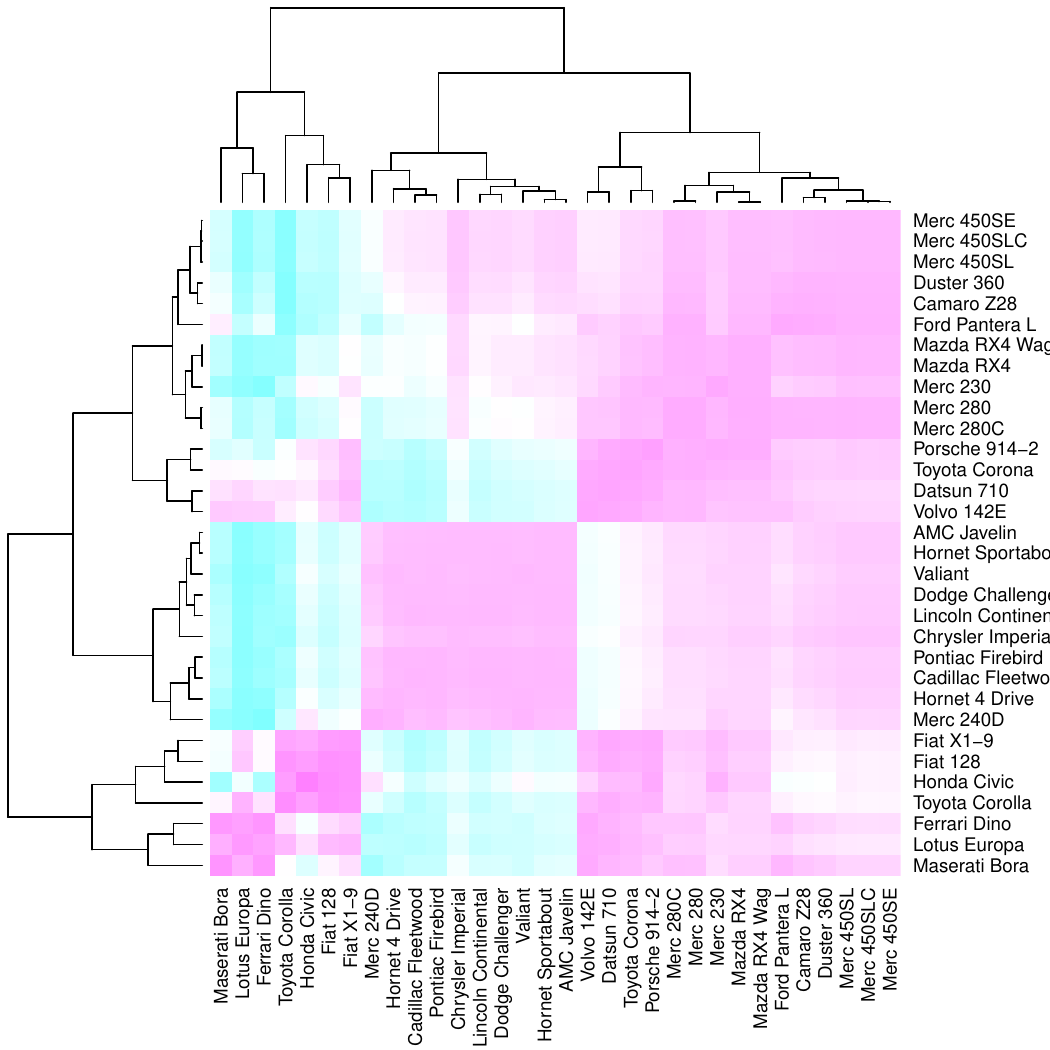}
\caption{Correlation matrix of points (cars).}\label{subfig:d}
\end{subfigure}
\begin{subfigure}[t]{0.9\textwidth}
\centering
\includegraphics[scale=0.3]{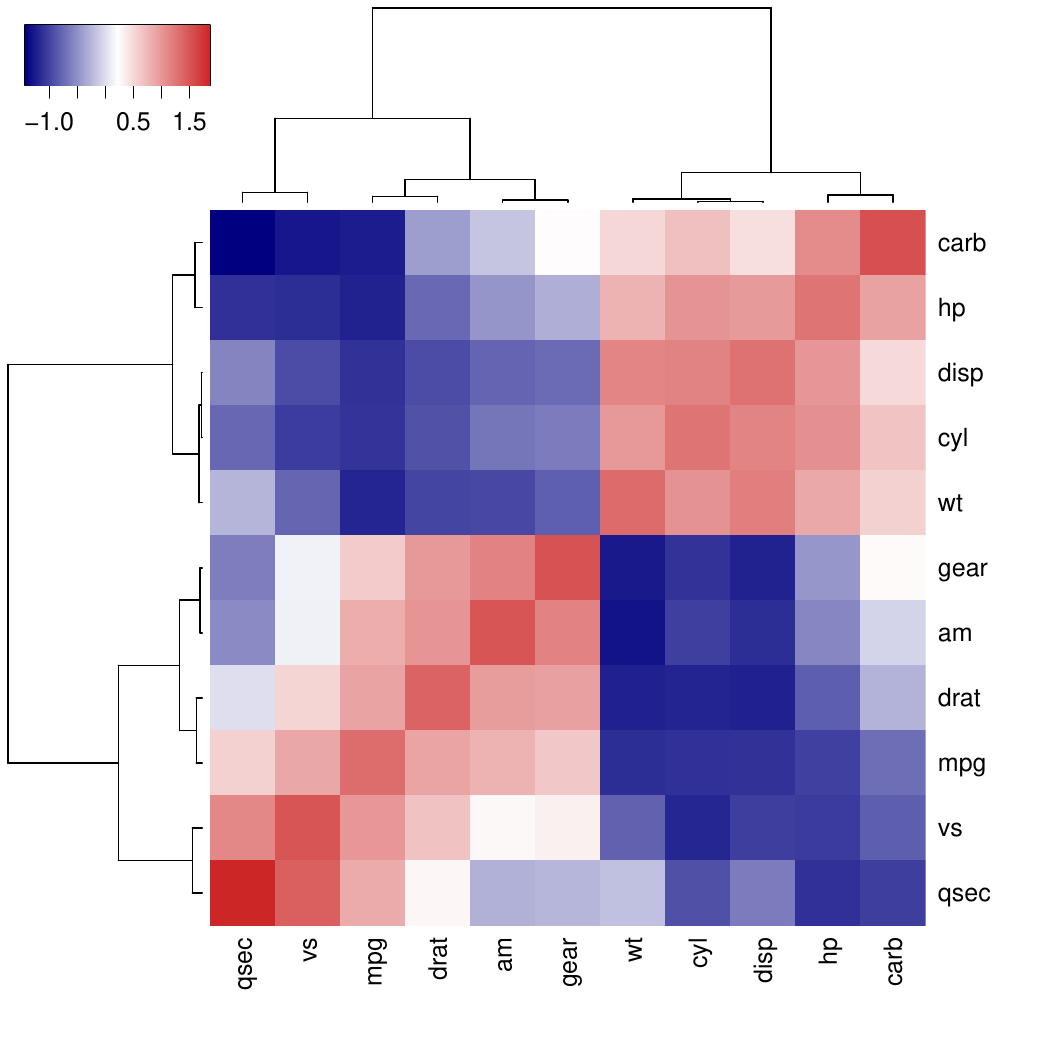}
\caption{Correlation matrix of the dimensions (attributes).}\label{subfig:c}
\end{subfigure}
\caption{The correlation matrices of rows and columns of the mtcars data set.}
\end{figure}

The diagram of \Cref{fig:heatmap} is drawn using R on the dataset mtcars (Motor Trend Car Road Tests) in R with dimension fuel consumption and 10 aspects of automobile design and performance for 32 automobiles (1973--74 models): mpg (Miles/(US) gallon), cyl (Number of cylinders), disp (Displacement (cu.in.)), hp (Gross horsepower), drat  (Rear axle ratio), wt Weight (1000 lbs), qsec (1/4 mile time), vs (Engine (0 = V-shaped, 1 = straight)), am (Transmission (0 = automatic, 1 = manual)), and gear (Number of forward gears). Dendrograms were used to sort the values of each dimension, independently.

Ideally, we want the cells of similar color to be adjacent to each other such that all the similar patterns are close together. In \Cref{subfig:b}, the attributes were sorted based on their correlation matrix. This is clearly not good enough since for example, "Maserati Bora" and "Ford Pantera L" have similar values on attributes "gear", "am" and "drat", but the differences in the shared values ("wt", "cyl", "disp", "hp", and "carb") does not allow the hierarchical clustering algorithm on the columns to detect them.
While a decision tree would have chosen one attribute per cluster, in the best case, a heat map merges the set of dimensions. A fully sorted heat map would allow us to see if there is a missing attribute, where a set of points belong to several clusters at the same time. In the example of~\Cref{fig:heatmap}, there are two main clusters: the first one with dimensions "wt", "cyl", "disp", "hp", and "carb", the second one with dimensions "vs", "mpg", "drat", "am", and "gear". However, there is a cluster with the attributes of the first cluster with additional dimensions "drat", "am", and "gear", a cluster with attributes "qsec" and "vs" that intersects cluster 2, and a subcluster of cluster 1 with attributes "wt" and "disp". These are not easily observed in \Cref{subfig:b}, but we see examples of these types of clusters in the experiments section.

Heat maps with rectangular clusters are a special case of the heat map sorting problem that we discuss.
A (mathematical) graph is a pair $G=(V,E)$ where $V$ is the set of vertices and $E$ is a subset of $V\times V\setminus \{\{v,v\}\mid v\in V\}$, called the set of the edges of $G$.
Hypergraphs are the generalization of graphs where each edge is a subset of (possibly more than $2$) vertices.
Using the adjacency matrix of hypergraphs $G=(V,E)$ as the heat map where the vertices are labeled based on their connected components, heat map sorting puts vertices with similar adjacency lists near each other, which means it finds subsets of vertices $V_1,V_2\subset V$ such that $V_1\times V_2\subset E$.

\subsection{Contributions}
In \Cref{sec:hardness}, we discuss the computational complexity of heat map sorting (HMS).
We prove that heat map sorting and some of its special cases are NP-hard by reductions from the Hamiltonian path and set cover problems. The special cases are where no dimension belongs to more than a cluster (dimension-disjoint), where having certain values for certain clusters means the point belongs to a cluster (rectangular clusters), and where no point or dimension belongs to more than one cluster (disjoint and dimension-disjoint rectangular clusters). We also discuss the hardness of HMS in MPC based on the connectivity problem.
HMS with disjoint and dimension-disjoint rectangular clusters is polynomial-time solvable.

In \Cref{sec:alg}, we give an approximation algorithm for a special case of heat map sorting and an exact exponential-time algorithm. We also give a heuristic algorithm for HMS.

In \Cref{sec:experiment}, we present and discuss the empirical results of running some of these algorithms on two graph datasets and compare them to $k$-means and DBSCAN clustering algorithms with LSH as the dimensionality-reduction method.
\section{Preliminaries}
\subsection{Problem Definition: Heat Map Sorting (HMS)}
Binary heat map sorting is defined as follows.
A binary matrix $M$, with $D$ as the labels of its rows and $P$ as the labels of its columns is given. A clustering based on $M$ is the set of clusters $C_i$, for $i=1,\ldots,k$ that are the connected components of the graph where two columns are adjacent if they are equal in one dimension and let $D_i$ be the set of rows with value $1$ in $M$ in the columns of cluster $C_i$.
Also, let $P_i$ be the set of points in $C_i$.
The number of preserved clusters is the maximum number of clusters that can be shown using a heat map.

In special cases, one of the sets $P_i$, $C_i$, or $D_i$ can be removed without changing the objective function:
\begin{itemize}
\item $P_i$: if the clusters are disjoint (no point belongs to more than one cluster),
\item $D_i$: if the clusters are dimension-disjoint (no dimension belongs to more than one cluster), and
\item $C_i$: if there is only one cluster.
\end{itemize}

The general version of heat map sorting uses several colors (unlike the binary heat map which is black and white). Each color is an instance of the binary heat map sorting, but the permutations have to simultaneously solve $\size{C}$ binary HMS problems. One way to convert the general case to the binary case is to use discretization using multiples of $\epsilon$, for an arbitrary constant $\epsilon>0$. This gives $1/\epsilon$ different values, and for each of these values, we solve a binary instance and return the one with the maximum number of clusters. This modification affects the cost of $k$-means and $k$-medians clustering problems by a factor at most $1+\epsilon$. Alternatively, we can add a new dimension for each of these discretized dimensions.

Given the ordering of the dimensions $D$, the ordering of the clusters is not uniquely defined, as it depends on the ordering of $P$. If the ordering of the subset of the clusters that are preserved in the solution is also given, then by sorting $P$ based on the order of the clusters, the solution can be retrieved. Let $S$ be the tuple that gives such an ordering of the clusters and their dimensions (the solution to the heat map sorting problem).

Let $\CC{(M)}$ denote the set of connected components of a graph with adjacency matrix $M$. Also, let $\pi:D\rightarrow D$ and $\sigma:P\rightarrow P$ be permutations of $D$ and $P$ respectively. Then, the output size of the heat map sorting is
\[
h=\min_{\pi,\sigma}\size{\CC({\pi(D)\times \sigma(P)})},
\]
and S=$(\pi,\sigma)$ gives the sorting.

\begin{definition}[Hypergraph Formulation of HMS]
Build a hypergraph for the heat map sorting problem by adding an edge for the set of dimensions in each cluster and a vertex for each dimension. Let $f$ be the maximum number of clusters that a dimension belongs to.
\end{definition}

\subsubsection{Rectangular Clusters}
If each cluster is the set of elements in the product of a range in each dimension of the matrix, we call the problem heat map sorting with rectangular clusters.
This happens, for example, when each dimension is clustered separately.
Clustering with rectangular clusters is also known as axis-parallel subspace clustering.

\subsubsection{Dimension-Disjoint Rectangular Clusters}
When the clusters consist of disjoint sets of dimensions, i.e., $\forall i,j, D_i\cap D_j=\emptyset$, we call the problem heat map sorting with dimension-disjoint clusters. We show that dimension-disjoint heat map sorting reduces to set cover in \Cref{lemma:still}.
Having clusters whose dimensions are disjoint means the order of clusters does not matter in this case.

\subsubsection{Maximizing The Jaccard Similarity to The Submatrix of HMS}
We use the Jaccard similarity of the submatrix of HMS after removing/merging the rows and columns to the original HMS as the measure of the quality of the solution. The formal definition is given in \Cref{def:density}.
\begin{definition}[Density Ratio]\label{def:density}
In the presence of outlier points and useless attributes, the heat map with clusters might be a smaller matrix $D'\times P'$. To measure the number of outliers and useless dimensions, we use the following formula, which we call the density ratio:
\[
\rho=\frac{\size{D'} \size{P'}}{\size{D} \size{P}}.
\]
\end{definition}
For the special case of two clusters that are considered as a binary classifier, this is the same as the accuracy of the model since it is the ratio of the number of correctly classified points to the number of points.

The decision version of explainable clustering can be reduced to heat map sorting by building a graph with the set of vertices that are points in $n$-dimensional space and the $i$-th dimension of the $j$-th point is $1$ if $i$ and $j$ are in the same cluster.
Minimizing the number of misclassified points $(o)$ in explainable clustering and maximizing the accuracy in heat map sorting are equivalent since the number of dimensions is always $k$ in the explainable clustering:
$
ok=\rho.
$
The accuracy and $k$ are in a trade-off with each other. By fixing the number of clusters to be $k$, we then minimize the number of outliers.

\subsection{Related Previous Work}
\subsubsection{Hardness of Clustering in MPC}
In set cover, a set of subsets $S_1,\ldots,S_k$ of a universe set $U$ is given and the problem asks for the minimum number of these subsets whose union is $U$.
Set cover in sublinear time requires a superlinear number of queries~\cite{indyk2018set}. This hardness result carries on to the MPC model.
Set cover parameterized with the size of the universe has a fixed-parameter algorithm~\cite{cygan2015parameterized} and Hamiltonian path parameterized with the clique cover (the minimum number of cliques in the graph that cover all the vertices of the graph)~\cite{lampis2018parameterized}.

Clustering problems $k$-means and $k$-medians in a metric space take a set of points and an integer $k$ and find $k$ of the input points as centers such that the distances from the points to their nearest centers are minimized, where the cost function of $k$-means is the $\ell_2$-norm (the square root of the sum of squares) of these distances and the cost of $k$-medians is the $\ell_1$-norm (the sum of absolute values) of these distances.

Single-linkage clustering takes a set of points and merges the closest pair of points at each step until reaching $k$ connected components. The tree of the merges in this algorithm (hierarchical clustering problems in general) is called a dendrogram.

Density-based clustering (DBSCAN) takes a set of points $P$, a radius $r$, and an integer $f$ and finds the connected components of the graph with a vertex for each disk of radius $r$ centered at a point of $P$ and an edge between two disks if they each have at least $f$ points and they intersect.
Each connected component is a cluster and points whose disks have less than $f$ points but intersect with a disk with more than $f$ points are considered as the members of the cluster of that disk.

Single-linkage and DBSCAN use connectivity, which is conjectured to be hard in MPC.

\subsubsection{Locality-Sensitive Hashing for Hamming Distance}
The Hamming distance of two strings is the number of positions where the characters differ. For binary vectors, we can convert them into strings and use the Hamming distance. Note that reordering the dimensions does not change the Hamming distance and $\ell_p$ distances.

Locality-sensitive hashing (LSH) is a method for solving approximate nearest neighbors in high dimensions. LSH methods that use random projections or take subsets of the dimensions also provide dimensionality reduction, in addition to preserving the distances. Locality-sensitive hashing for Hamming distance using a subquadratic number of subspaces exist~\cite{gionis1999similarity,indyk1998approximate,datar2004locality,andoni2014beyond}.
This method has been used in hierarchical clustering before~\cite{cochez2015twister}.

For a $c$-approximation using the LSH for Hamming distance for $n$ points~\cite{indyk1998approximate}, a set of $n^{1/c}$ subspaces with one dimension must be chosen uniformly at random and the points with the same projection in these subspaces are considered to be approximate near neighbors. This algorithm has a constant probability of success which can be improved to a high probability by sampling $\ln n$ dimensions in each subspace.

\section{The Computational Complexity of HMS}\label{sec:hardness}
We show that heat map sorting is NP-hard, even when the clusters do not share dimensions.

\subsection{Binary HMS}\label{sec:nphard}
For $\size{D}=O(1)$, the problem is polynomial-time solvable since checking all permutations of the dimensions $\size{D}!=O(1)$ would take polynomial time. Also, if $k=O(1)$, all permutations of clusters can be checked, which takes $k!=O(1)$ time.

\begin{theorem}\label{theorem:hard}
Heat map sorting is NP-hard.
\end{theorem}
\begin{proof}
We reduce the Hamiltonian path problem to heat map sorting. Let $G'=(V',E')$ be the graph from an instance of the Hamiltonian path problem. Define an instance of the heat map sorting using matrix $G$ as follows:
The rows are the set of dimensions $D=\{(v,u,w)\mid (v,u)\in E', (u,w)\in E'\}$, the columns are the set of points $P=V'$, and there is a $1$ at row $(v,u,w)$ only in column $u$.
Each permutation of the columns is a permutation of the vertices, so, one of them is the Hamiltonian path in $G'$, if it exists. The number of dimensions $(\size{D})$ is the number of paths of length $2$ in $G$'.

The heat map sorting problem maximizes the number of connected components that are preserved. A connected component of $G$ is any Hamiltonian path in each connected component of $G'$ because the longest sequence of triples is the triples of consecutive vertices in the Hamiltonian path. Since there is only one connected component in $G'$, then, the Hamiltonian path also maximizes the number of clusters that are preserved. If $G'$ does not have a Hamiltonian path, then, it is easy to check that the sequence of vertices that maximizes the number of clusters of $G'$ that are preserved in $G$, which is given in the order of the columns of $G$, does not form a path.
\end{proof}

\subsection{HMS with Dimension-Disjoint Clusters}

\begin{theorem}\label{lemma:still}
Deciding whether a heat map sorting instance with dimension-disjoint clusters preserves at least $k$ clusters assuming it has to cover all the points is equivalent to set cover and therefore NP-complete.
\end{theorem}
\begin{proof}
Given the heat map sorting, checking if all the points are covered can be done by checking if all the points that belong to the same cluster appear together. This takes $O(kdn)$ time, where $n$ is the number of points, $d$ is the number of dimensions, and $k$ is the number of clusters.

We reduce the set cover problem to heat map sorting.
Define a column for each of the elements of the universe as the input points in the heat map sorting problem. For each set of the set cover problem, define a cluster. Define the dimensions as the maximal subsets of clusters that contain the same set of points. A solution to the heat map sorting problem finds a subset of $k$ clusters that can cover all the input points. This is a solution to the set cover problem with $k$ sets (we need to modify the condition at most $k$ clusters to exactly $k$ clusters, which is easy as the decision for all values $1,2,\ldots,k$ suffices for this reduction, which makes it a polynomial-time reduction).

To show the problems are equivalent, we do the reverse reduction, too.
Since the dimensions are disjoint, in the set of rows for each cluster, some of the columns are in the cluster and some of them are not. This means there is a set of intervals in each ordering of the points (columns) that show the members of each cluster. The goal is to find an ordering that preserves at least $k$ clusters. Now, we sort the matrix in two steps. Sort the rows based on the cluster number to put the dimensions of each cluster adjacent to each other. Then, sort the columns (points) based on these new dimensions. For any maximal subset of clusters that have some points in common but not all points, create a set in the set cover problem. A solution to the set cover problem solves heat map sorting.
\end{proof}

A corollary of \Cref{lemma:still} is \Cref{lemma:preserve}.
\begin{lemma}\label{lemma:preserve}
The reduction of \Cref{sec:disjoint} from heat map sorting with dimension-disjoint clusters to set cover is approximation factor preserving.
\end{lemma}

Since the approximation factor is preserved in this reduction, the $O(\log n)$-approximation algorithm for set cover and the $f$-approximation for set cover where each element appears in at most $f$ sets (maximum degree $f$ in the graph representation) give the same approximation ratios for heat map sorting.

\subsubsection{HMS with Disjoint and Dimension-Disjoint Clusters}\label{sec:seq}
Based on the reduction from dimension-disjoint clusters to heat map sorting, $f$ is the number of clusters that include a dimension. If in addition to clusters being dimension-disjoint, the points also belong to at most one cluster (the clusters are also disjoint), then $f=1$. So, the linear programming rounding algorithm gives an exact solution.
This gives \Cref{lemma:p}.

\begin{theorem}\label{lemma:p}
Heat map sorting with clusters that are both disjoint and dimension-disjoint is solvable in polynomial time.
\end{theorem}

However, linear programming is P-complete and we need another algorithm for the parallel model. Fortunately, any maximal solution would give an exact solution for this case.
To compute a maximal solution, we fix one of the points with a high degree in terms of $k$ as the first column $c_1$, and then at each step $i$, $(i>1)$, we add the point that is most similar to $c_{i-1}$ and remove it from the set.

\subsubsection{Hardness of HMS in MPC}
Verifying the solution of a heat map sorting instance requires computing the connected components in MPC. The problem of computing the connected components in MPC using a constant number of rounds is open~\cite{yaroslavtsev2018massively}. Since the adjacency matrix of a hypergraph can be defined similarly to a heat map, computing the connected components in MPC reduces to verifying a heat map sorting.
We formalize this in \Cref{theorem:mpc}.
\begin{theorem}\label{theorem:mpc}
Verifying a given heat map sorting is at least as hard as st-connectivity in MPC. 
\end{theorem}

\section{Algorithms}\label{sec:alg}

\subsection{An Approximation Algorithm for Dimension-Disjoint HMS}\label{sec:disjoint}
If we assign the points to their lexicographically first  and merge the clusters that share at least one dimension,  we get an instance of disjoint clusters that are also dimension-disjoint.
\begin{theorem}\label{theorem:dimension}
If the maximum number of clusters that share a dimension is $f_1$ and the maximum number of clusters a point belongs to is at most $f_2$, then a maximal solution is a $f_1f_2$-approximation.
\end{theorem}
\begin{proof}
Based on \Cref{lemma:preserve}, the approximation ratio of the set cover algorithm is preserved, which is $f_1$ for the matrix $\pi(D)\times \tau(\{C_1,\ldots,C_k\})$ and $f_2$ for the matrix $\pi(D)\times \sigma(P)$, where $\tau: \{C_1,\ldots,C_k\} \rightarrow 2^{\{C_1,\ldots,C_k\}}$ is the merging of clusters.

If the algorithm chooses the worst order of covering the clusters and the dimensions, each of them is at most repeated $f_1$ times and $f_2$ times, respectively. So, each dimension repeats at most $f_1$ clusters and each cluster repeats at most $f_2$ points, which results in at most $f_1f_2$ repeats of the optimal solution:
\[
\lvert\{ \pi (D)\mid \pi : D\rightarrow D \}\rvert = f_1,\quad \lvert \{ \sigma(P) \mid \sigma: P\rightarrow P\} \rvert = f_2\]
\[ \Rightarrow \size{\CC(\pi(D)\times \sigma(P))} \leq f_1f_2 h.\qedhere
\]
\end{proof}
\subsection{A FPT Algorithm for HMS}
The goal is to find a permutation of the dimensions $(D)$ such that the maximum number of intervals are preserved, where the intervals are the set of consecutive dimensions (rows of the matrix) that contain each cluster.
This problem can be formulated as finding the longest path in the hypergraph representation of heat map sorting: one vertex for each dimension and a hyperedge for each subset of dimensions $(S_1,\ldots,S_k)$.

Let $h$ be the number of vertices of degree more than $2$ in the hypergraph and let $\psi=\max_{v\in D} \deg(v)$, where $\deg(v)$ is the degree of vertex $v$ in the hypergraph.

A sketch of the algorithm has the following four phases:
\begin{enumerate}
\item Break the sets $S_1,\ldots,S_k$ into smaller maximal subsets (the set $P$) such that each element in a set $p\in P$ intersects the same subset of $S_1,\ldots,S_k$.
\item Build a directed-acyclic graph (DAG) starting from the sets of $P$ that do not contain another set of $P$ as the first level and merge them to build the rest of the sets in $P$.
\item Assign weight $1$ to nodes of the sets in $S_1,\ldots,S_k$ and $0$ to the rest of the nodes of the DAG.
\item Find a set of disjoint binary trees with the maximum total weight by enumerating all cases of choosing two of the children of each node.
\end{enumerate}
\Cref{alg:longest} gives the detailed version.


The time complexity of the first phase is $O(d^2\log d)$ for sorting the sets $I_u, \forall u\in U$ and building a tree by merging them.
Building a DAG by merging can be done by a breadth-first search, which takes $O(d)$ time. The third step takes $O(d^2)$ time to compare the equality of the sets and assign the numbers based on the sorted order of their size from the first step.
The time complexity of the fourth step is:
\begin{align*}
\sum_{i=1}^{\size{P}} \binom{\size{P}}{i} i &= \sum_{i=1}^{\size{P}} \frac{1}{\size{P}+1} \binom{\size{P}+1}{i-1}
=\frac{1}{\size{P}+1} \sum_{i'=1}^{\size{P}-1} \binom{\size{P}+1}{i'}\\
&= \frac{1}{\size{P}+1} (\sum_{i'=1}^{\size{P}+1} \binom{\size{P}+1}{i'}- \binom{\size{P}+1}{\size{P}+1}-\binom{\size{P}+1}{\size{P}})\\
&=\frac{1}{\size{P}+1}  (2^{\size{P}+1}-1-(\size{P}+1))=\frac{2^{\size{P}+1}-1}{\size{P}+1}-1 = O(\frac{2^{\size{P}}}{\size{P}}).
\end{align*}

For $d\leq n$, the total time complexity of the algorithm is $O(n^2\log n+\frac{2^{\size{P}}}{\size{P}})$. So, the algorithm is fixed-parameter tractable in terms of $\size{P}$. The assumption that $d\leq n$ is natural since removing repeated dimensions can be done in linear time and is parallelizable (this happens because of the pigeonhole principle as the number of dimensions is more than the number of points).

\begin{center}
\begin{minipage}{\linewidth}
\begin{algorithm}[H]
\begin{algorithmic}[1]
\Procedure{MaxTree}{tree $F$, DAG $T$ with weights $w$}
\If{$T=\emptyset$}
\State{Return $0$}
\EndIf
\For{$v\in T$}
\State{A=\Call{MaxTree}{$F\cup \{v\}$,$T\setminus \{v\}$}+w(v)}
\State{B=\Call{MaxTree}{$F$,$T\setminus \{v\}$}}
\If{$A>B$}
\State{Return $A$}
\Else
\State{Return $B$}
\EndIf
\EndFor
\EndProcedure
\Require{A set of sets $S_1,\ldots,S_k$ of a universe set $U$}
\Ensure{A permutation of the elements of $U$}
\State{$S'_i=S_i, \forall i=1,\ldots,k$}
\State{$U=\cup_{i=1}^k S_i$}
\State{Create a DAG $T$ with one node for each element of $U$.}
\While{$\exists S'_i, \size{S'_i}>1$}
\State{Add one vertex to $T$ for each $S'_i$.}
\State{$U'=\cup_{i=1}^k S'_i$}
\State{Connect the nodes for each $S'_i$ in $T$ to the nodes for the elements of $U'$ in $T$.}
\For{$u\in U'$}
\State{$I_u=\{i\mid S'_i\ni u\}$}
\EndFor
\State{Hash the elements of $U$ based on $I_u, u\in U'$.}
\For{each bin $b$ in the hash table}
\State{Merge the elements in $b$ into a single node in $U'$.}
\EndFor
\State{Update the sets $S'_i$ based on the merged sets.}
\EndWhile
\For{$i=1,\ldots,k$}
\State{Find $S_i$ in $T$ and assign weight $1$ to it.}
\EndFor
\State{Return \Call{MaxTree}{$\emptyset$,$T$}}
\end{algorithmic}
\caption{A FPT Algorithm for HMS}
\label{alg:longest}
\end{algorithm}
\end{minipage}
\end{center}
\subsection{HMS for Disjoint Rectangular Clusters with Outliers}
We give an algorithm for disjoint rectangular clusters in multi-color HMS in the presence of outliers using existing clustering algorithms (DBSCAN, $k$-means, and single-linkage clustering) and some basic operations (loops, sequence of executions) and database operations (join, select, aggregation).

These are the main components of the algorithm:
\begin{itemize}
\item
DBSCAN with Euclidean distance: It finds the connected components of the rows (equivalently, the columns) and removes singleton clusters. The Euclidean distance gives the mean-squared error between two vectors, if the dimensions are normalized, which is the case with heat maps.

This does not change the connected components of the graph after removing outliers.
\item
$k$-means with cosine distance: When there are only $k$ clusters, there can be $k$ merged dimensions. The cosine distance is the same as the dot product of the vectors. This is used to merge the dimensions and group the points together.

If two vectors have nearly equal elements, their mean-squared error is smaller than other vectors. So, using cosine distance would not remove them. Using $k$-means would not change correlated attributes, assuming they are far enough from the noise (outliers).
\item
Single-linkage clustering to initialize DBSCAN: As the input of the problem is $k$, finding the minimum radius for which the clusters are separated can be solved using single-linkage clustering. This requires running the rest of the algorithm only once, while discretization of the range of values takes a constant number of rounds and it takes $\epsilon$ as an input parameter, so, setting $\epsilon$ would require something like this (or the user sets the value based on the data).
\end{itemize}

The algorithm runs three steps for the rows and then the columns. These steps are as follows. First, the algorithm uses single-linkage the initialize DBSCAN, then, it runs DBSCAN and merges the dimensions or points using $k$-means.
These three steps can be repeated if doing so would improve the accuracy.

When the distances between the points inside the clusters are less than points in different clusters and outliers, the algorithm is equivalent to the algorithm for finding rectangular clusters in binary HMS.
\section{Empirical Results}\label{sec:experiment}
Explainable clustering is best compared using figures. So, we implement our algorithms and compare their outputs based on both the heat map and the accuracy of the model. As we needed to compare our work to existing methods, we used two accuracy measures: one for the ratio of the correctly classified points to the number of input points and another one for the ratio of the remaining clusters to the number of input clusters (colors).

The experiments are in the form of case studies where several algorithms are implemented on the same dataset one after another. The first dataset has several kilos of points and disjoint rectangular clusters (disjoint and dimension-disjoint clusters). The second dataset is roughly the same size (several kilos of points) and has various types of intersecting clusters (intersecting and nested). The number of dimensions is almost the same as the number of points (medium-size data in a medium number of dimensions).

After explaining the implementations and results for each algorithm in this paper, we compare our results with existing algorithms such as $k$-medoids and DBSCAN. Since they only cluster the points and not the dimensions, we use LSH as a dimensionality reduction tool.

The goal is to compare the power of these models to express clustered data, so, all the input data is labeled. The algorithm that preserves more labels is better.

Note that the figures representing point sets have a lower resolution than the actual data for better visualization and a smaller file (the data was rounded to reduce its size).

\subsection{Datasets, Hardware, and Programming Languages}

\textbf{Datasets:}\\
We use several graphs from the SNAP graph datasets~\cite{snapnets}.

\textbf{Hardware:}\\
The hardware specifications are as follows:
\begin{itemize}
\item CPU: Core i7 2.90GHz processor with 16 threads,
\item Memory: 16GB RAM, and 2TB disk.
\end{itemize}

\textbf{Programming Languages:}\\
The languages used are as follows:
\begin{itemize}
\item
For indexing (mapping), backtracking, and algorithms involving more than a constant number of elements we used C++.
\item
Big data processing including partitioning, filtering using regular expressions, joining datasets with their labels, and sorting and re-indexing (shuffling) was implemented using Linux shell commands by setting the block sizes based on the memory limit of MPC.
\item
For smaller datasets, Microsoft Excel formulas were used in some cases to index points and compute the maximum and minimum values.
\item
Figures were specified as commands and drawn using TikZ package and complied with LuaLaTeX, except for \Cref{fig:tree} which was drawn using Ipe drawing tool.
\end{itemize}

\FloatBarrier
\subsection{Verifying Disjoint and Dimension-Disjoint HMS in MPC}\label{sec:email}
Here, we solve the decision/verification version where the solution is given and the goal is to compute the cost. Disjoint and dimension-disjoint (disjoint rectangular clusters) is the polynomial-time solvable special case of HMS.
\subsubsection{Using The Approximation Algorithm}
The algorithm first indices the vertices based on the spatial index of their heat map matrix. Using these indices, the values of the same block become adjacent. Then, the algorithm counts the number of ones in blocks of size $\sqrt{m}\times \sqrt{m}$. Finally, the algorithm selects the blocks with at least $m/2$. We relaxed the condition of being an exact rectangle (submatrix) to more than a triangle to cover missing points and directed graphs.

This algorithm takes $O(1)$ rounds in MPC since the sum of blocks can be computed locally and filtering them based on a given threshold or sorting them also takes $O(1)$ rounds.

We use email-Eu-core network~\cite{leskovec2007graph,yin2017local} dataset. It is a directed graph with 1005 vertices and 25571 edges.
The labeling of the vertex-disjoint clusters using $42$ labels is also given.
For easier implementation, we use $10^{\lfloor\log_{10}\sqrt{m}\rfloor}$ as the side length of the squares of the blocks.

\Cref{fig:email} shows the whole input.
\Cref{fig:cluster} shows the blocks with at least 50 non-zero values among $10\times 10$ blocks. This gives blocks 22, 23, 32, and 33 with counts 60, 57, 60, and 73. These blocks map to the vertices $\{20,21,\ldots,39\}$.
\begin{figure}
\centering
\begin{subfigure}[b]{0.5\textwidth}
\centering
\includegraphics[scale=0.06]{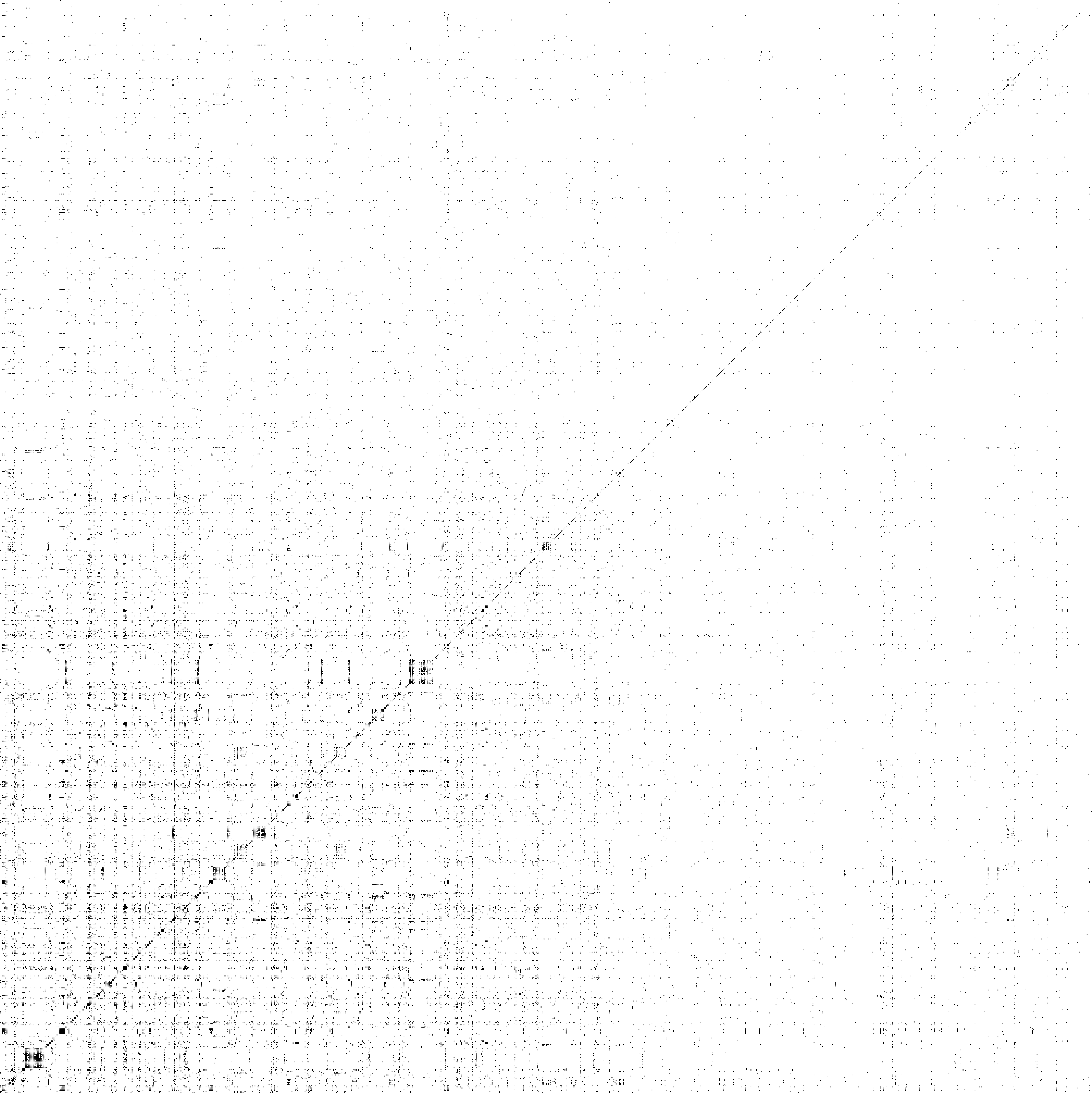}
\caption{The heat map of the adjacency matrix of the graph.}\label{fig:email}
\end{subfigure}\;
\begin{subfigure}[b]{0.35\textwidth}
\centering
\includegraphics[scale=0.1]{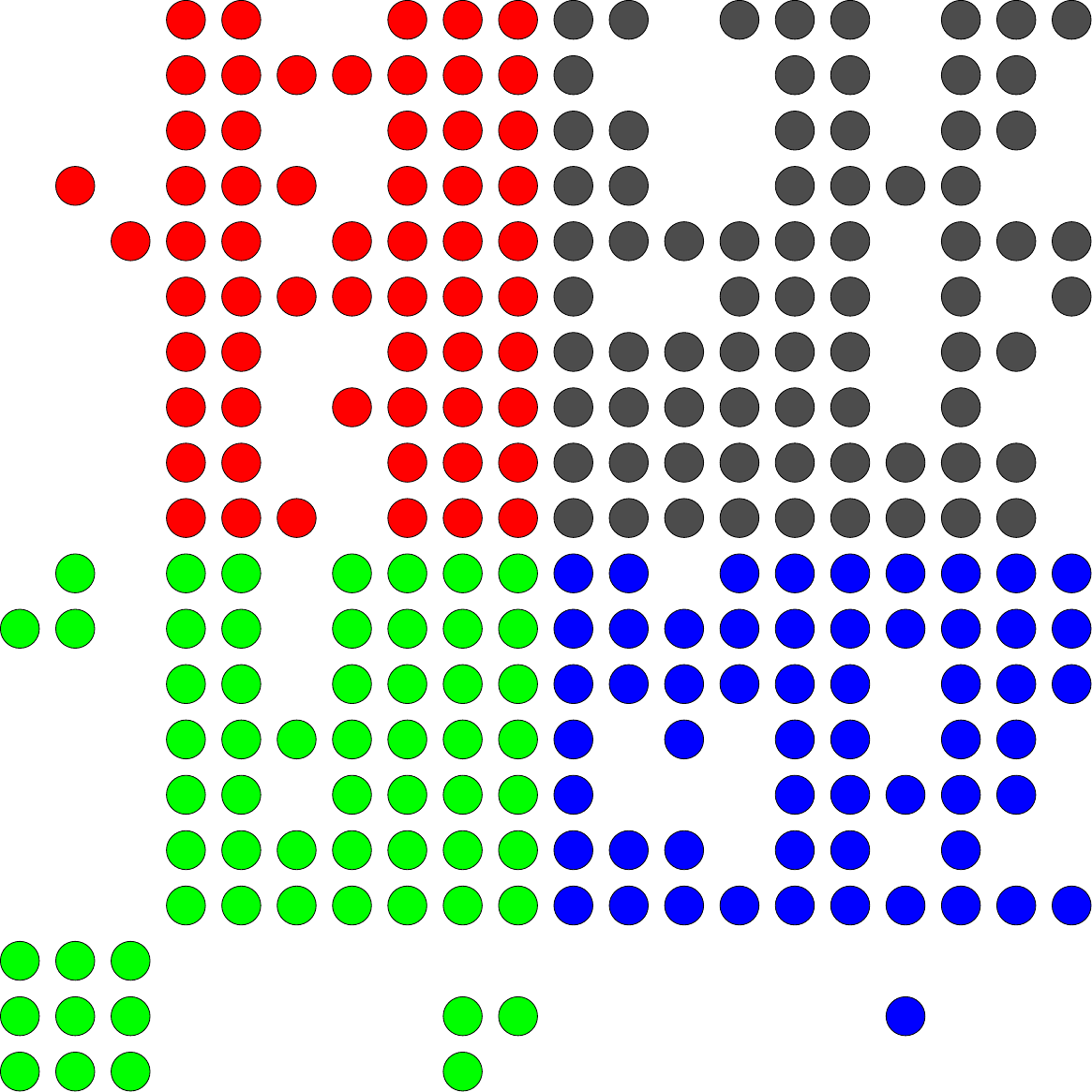}
\caption{A cluster expanding $4$ blocks, and some outliers. Each block is represented by a different color.}\label{fig:cluster}
\end{subfigure}
\caption{Representation of email-Eu-core network}
\end{figure}
We can compute the density ratio of the matrix in its original ordering and using blocks of size $10\times 10$, which is a $100$-approximation of $\rho$:
\[
\rho\approx \frac{4}{6005}\approx 0.07\%.
\]
The same clusters using individual cells instead of blocks have density ratio:
\[
\rho \approx \frac{249}{25571}\approx 0.97\%.
\]
Because of the blocking, this is still a $4$-approximation (a $2$-approximation in each dimension because the range is rounded to the nearest block size).
%
\subsubsection{Using The Exact Algorithm}\label{sec:email2}
Consider the graph from \Cref{sec:email}.
Since $(42!)^2\approx 2 \times 10^{102}$, we cannot check all the permutations of clusters. So, we prune away clusters that have a few members.

Based on the labels of the dataset, there are 6 clusters with at least 50 points, which are 14, 15, 1, 7, 4, and 21 (shown in \Cref{fig:before}).
These clusters have $6113$ points in total.
For these clusters, there are $(6!)^2=518400$ permutations of the dimension-point matrix. We also remove the vertices that have no outgoing edges. \Cref{fig:after} shows the same clusters after sorting. We see that each cluster has a dense rectangular region. We also see that there are other clusters in the intersection of other clusters that were not marked in the input. Since these are research email exchanges, these might be interdisciplinary research conducted by researchers from several departments.

\begin{figure}[h]
\centering
\begin{subfigure}{0.45\textwidth}
\centering
\includegraphics[scale=0.6]{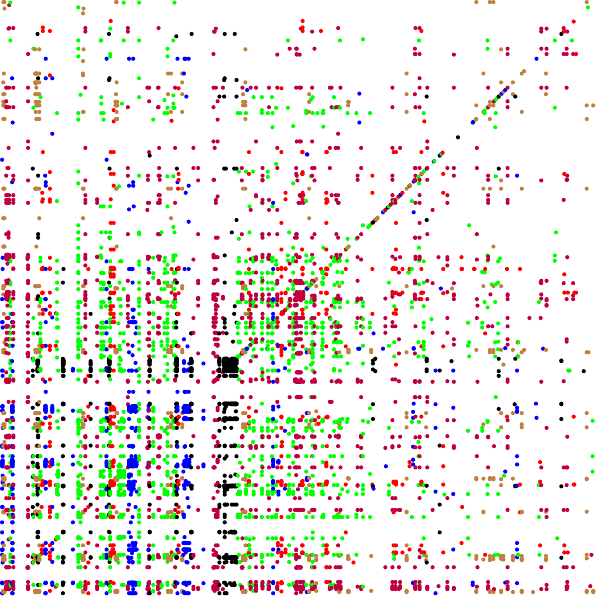}
\caption{The heat map of the selected clusters before sorting. Each color represents a different cluster.}\label{fig:before}
\end{subfigure}
\begin{subfigure}{0.45\textwidth}
\centering
\includegraphics[scale=0.013]{fig5v2.pdf}
\caption{The heat map of the selected clusters after sorting in the order $(14,15,1,7,4,21)$.}\label{fig:after}
\end{subfigure}
\end{figure}

To compute $\rho$, we compute the rectangle containing each cluster and consider all the points that do not lie inside a rectangle as outliers.
In the sorted order of points based on their clusters, we draw a separating line such that the number of points from the current cluster dominates the number of points from other clusters in each dimension. First, we count the number of points per cluster for each value of the first coordinate that appears in the input and find the cluster with the maximum value for each one. Then, we take the median of several consecutive cluster numbers (here we used $11$) to make sure there is (at most) one interval per cluster. Repeating the same process for the other dimension gives us disjoint intervals and the clusters look like squares.
This works because the clusters are almost disjoint, dimension-disjoint, and rectangular.

After reordering the clusters, there are $5002$ points inside the clusters.
So, the value of the density ratio for this subset of the data is:
\[
\rho =\frac{5002}{6113} \approx 81.83\%.
\]
\FloatBarrier
\subsection{The Output of The FPT Algorithm for HMS in MPC}\label{sec:cisco}
We repeat the experiments using Cisco networks dataset~\cite{madani2022dataset}.
One of the graphs from the set labeled extra graphs of the dataset is shown in \Cref{fig:cisco}.
We choose the graph of the network traffic of day 2022/10/13 from the Cisco network dataset.
The data for each graph contains the edge list (start and end vertices of directed edges), and metadata such as port and protocol for each edge. It has $59739$ edges and $6973$ vertices. The $7709$-th edge of the file for day 2022/10/13 happens for edge $(19548,19839)$ that belongs to all the $42$ clusters, which is the maximum among all other edges. So, all the clusters intersect and $f_2=42$. The approximation ratio is therefore $42$.
We consider the first protocol in the list of an edge as its cluster, so $f_1=1$.
Using this modification, each point belongs to one cluster.
\begin{figure}
\centering
\includegraphics[scale=0.03]{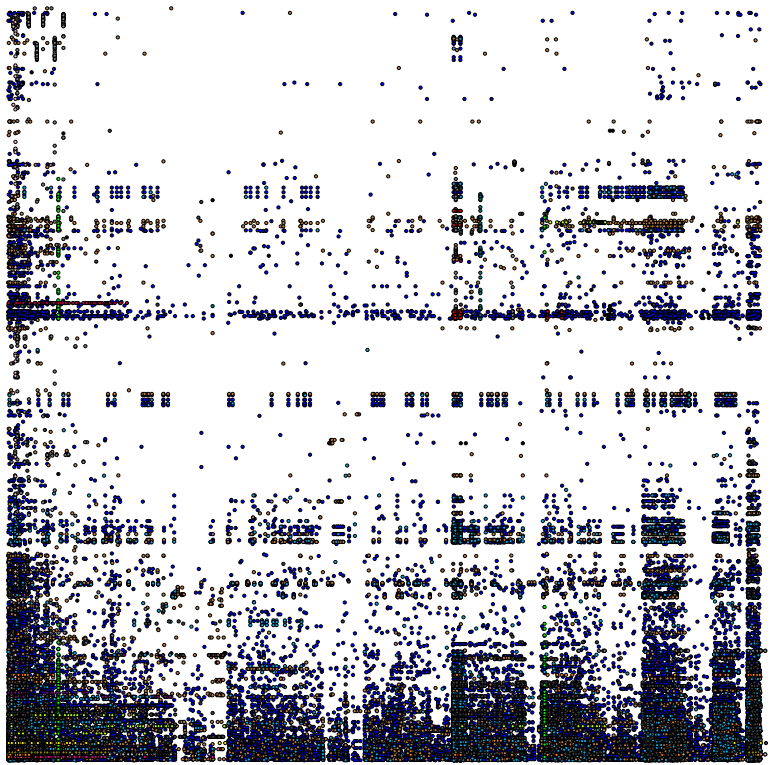}
\caption{The head map of the adjacency matrix of the network traffic graph for 2022/10/13. Different colors represent different protocols.}
\label{fig:cisco}
\end{figure}

\Cref{fig:ciscos} shows a maximal solution to the heat map sorting problem.
\begin{figure}
\centering
\includegraphics[scale=0.07]{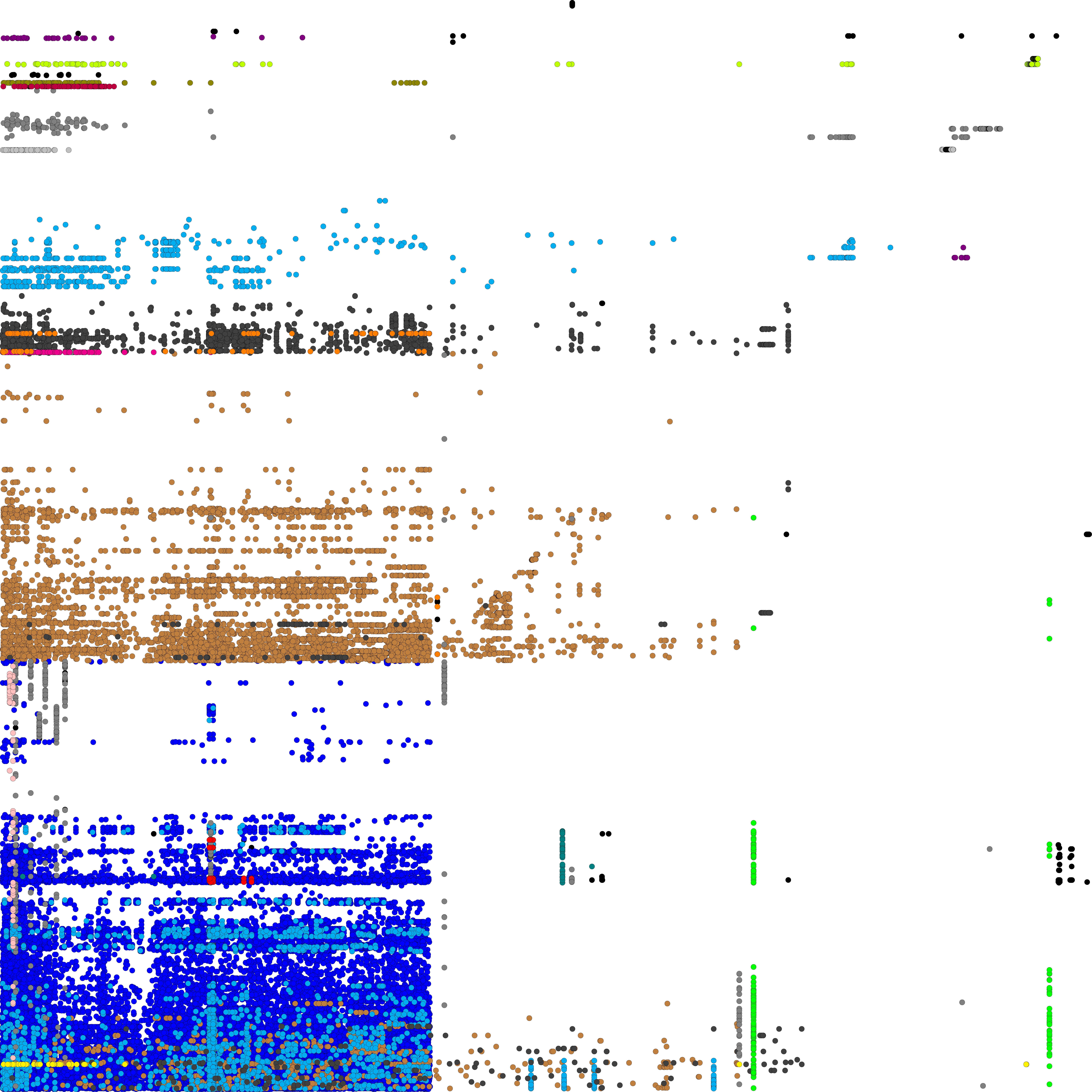}
\caption{The head map of \Cref{fig:cisco} sorted using protocols as labels.}
\label{fig:ciscos}
\end{figure}
From \Cref{fig:ciscos}, it is clear that the clusters are intersecting.
To get disjoint and dimension-disjoint rectangular clusters, we have to merge some of the intersecting clusters. Using the set of clusters in the input, we can use a frequent itemset algorithm to check whether the merged clusters actually conform to the input data.

There are $57$ subsets of clusters of size at least $2$ that appear in this dataset. \Cref{fig:tree} shows the values for $6$ of such clusters. The cluster $\{2,1\}$ cannot be added to this ordering because the capacity two for the neighbors of $1$ is filled by $3$ and $4$. Alternatively, we could have added $2$ in place of $4$ to keep cluster $\{2,1\}$ and lose cluster $\{4,1\}$ instead.
\begin{figure}
\centering
\includegraphics[scale=0.6]{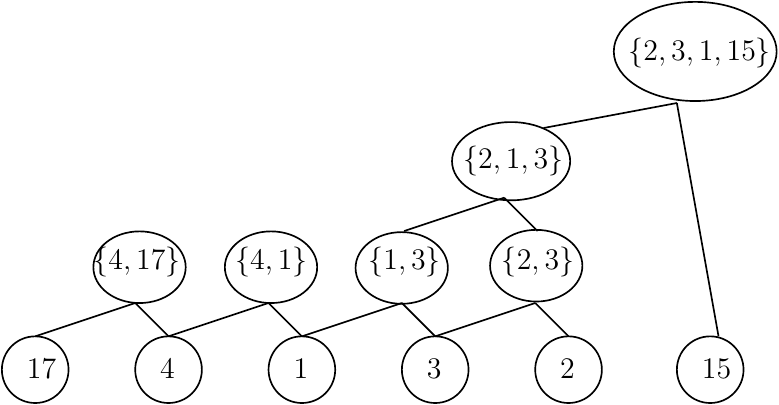}
\caption{The tree of the optimal ordering of clusters $1$,$2$,$3$,$4$,$15$,$17$,$\{4,17\}$,$\{4,1\}$,$\{1,3\}$,$\{2,3\}$,$\{2,1,3\}$,$\{2,3,1,15\}$.}
\label{fig:tree}
\end{figure}

After reordering the points according to \Cref{fig:tree} which is $17,4,1,3,2,15$, we arrive at \Cref{fig:box} for clusters $1$,$2$,$3$,$4$,$15$,$17$,$\{4,17\}$,$\{4,1\}$,$\{1,3\}$,$\{2,3\}$,$\{2,1,3\}$,$\{2,3,1,15\}$. While these clusters are not disjoint, we were able to represent them using a set of boxes without including or excluding the points of those clusters.
\begin{figure}
\centering
\includegraphics[scale=0.16]{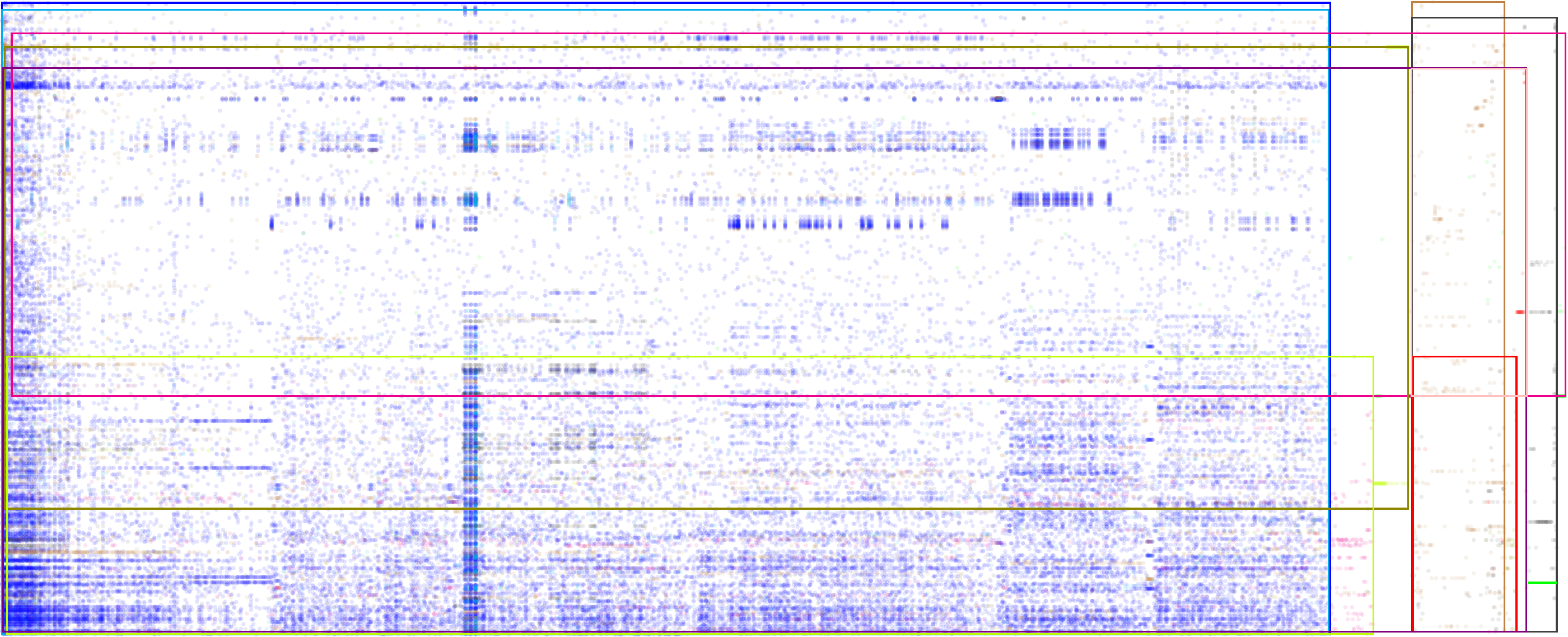}
\caption{The set of clusters from \Cref{fig:tree} and their bounding boxes after rearranging the values of the axes.}
\label{fig:box}
\end{figure}
\subsection{HMS with Outliers for Rectangular Clusters in MPC}\label{sec:disjoint}
Consider the subgraph of the tree from \Cref{sec:cisco} for clusters $1,2,3,4,15,17$ and their supersets. We find the subset of clusters that is most likely to contain each $x$-coordinate value based on the number of points with that $x$-coordinate in each cluster. The number of points becomes $57913$. Note that we have assigned each point to exactly one protocol (original clusters) arbitrarily (this removes protocol number $15$). The intervals for $x$-values (source IP) are shown in \Cref{fig:interval}. Clusters extracted by this method are different from the ones in the input, which means only the traffic of a subgraph (subnet) of the network can be predicted using this model.
\begin{figure}
\centering
\includegraphics[scale=0.6]{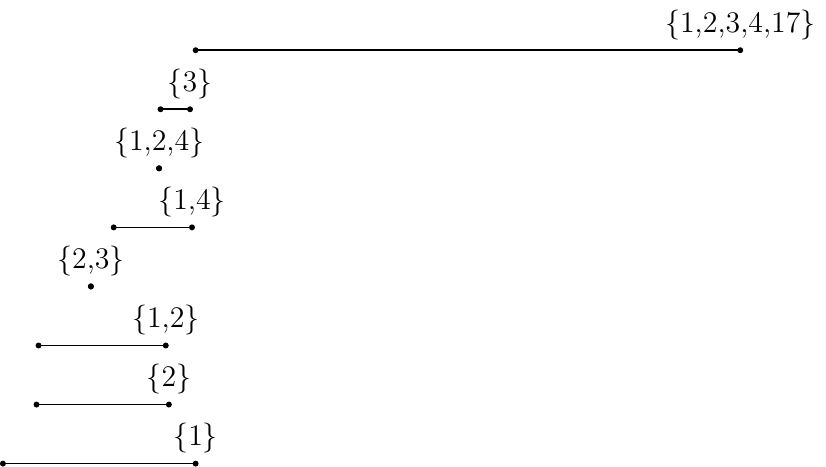}
\caption{The extracted intervals of $x$-coordinates of clusters ($y$-coordinates) that are subsets of $1,2,3,4,17$.}
\label{fig:interval}
\end{figure}

We take the median of the cluster numbers of $11$ consecutive numbers around each value for $x$ and $y$ separately with repetitions (do not remove repeated values that belong to different points) and if it is at least $0.9$ of the maximum number of elements per cluster at that coordinate we keep it, otherwise, we declare the point as an outlier. Considering outliers using points with repetitions guarantees that thin clusters (take a small range in one of the dimensions) are preserved. Note that the interval size for computing the medians ($11$) depends on the minimum cluster size ($5$), and the minimum ratio of the size of the intersection of two clusters to their maximum size ($0.9$) depends on the sizes of the clusters. Also, the ranges of outliers are defined for each cluster separately, so, if a point falls inside the range of another cluster, it is considered an outlier.

Consider the subset of points in clusters $1,2,3,4,17$ and other clusters intersecting with them. There are $57913$ points in this subset, $33906$ of them are outliers based on the disjoint ranges specified for clusters (see \Cref{fig:interval} for an example of such intervals). So, the density ratio using disjoint clusters and the same ordering as \Cref{fig:ciscos} is:
\[
\rho=\frac{24007}{57913}\approx 41.45\%.
\]
Based on the large gap between $\rho=100\%$ of intersecting clusters from \Cref{sec:cisco} and $\rho=41.45\%$ of disjoint clusters in this section, it is clear that our method significantly outperforms partitioning-based clustering algorithms on some datasets but not others (\Cref{sec:email,sec:email2}).
\subsection{Dimensionality Reduction via LSH}\label{sec:dim}
For heat map sorting, we need an algorithm that is not sensitive to the ordering of the dimensions and points. So, we use LSH for Hamming distance with the set of the columns of the matrix as the input.
Testing all possible arrangements of the dimensions (or clusters) takes too long. Using the dimensionality reduction of \Cref{sec:dim}, the dimension decreases, so, it can be easily checked.

Using the Cisco network dataset as the input, we have $n=59739$, so, the number of subspaces in each experiment must be about $10$ to guarantee a good performance when computing all the permutations $(10!\approx 3.6\times 10^6)$. By setting the approximation ratio $c=5$ for the distances between points, the number of subspaces must be at least $n^{1/c}=59739^{1/5}\approx 10$ and the number of dimensions of each subset must be at least $\ln n\approx 11$.

Using $10$ subspaces each with $11$ dimensions, the number of clusters with at most $11$ dimensions is $92$ and they have $59721$ points of the input. So, at most $92\%$ of the clusters and $\rho=\frac{59721}{59739}\approx 99.97\%$ of the points can be preserved in this method. This value is for the best case. When using random sampling, to cover $d$ dimensions, we need to sample $d\log d$ points based on the bound from the coupon collector problem, so, the maximum number of dimensions is $6$ since $6\ln 6\approx 11$.

Empirically, using a single run, we found the value to be $31$ clusters and $51578$ points with clusters of dimension at most $2$. So, $31\%$ of the clusters and $\rho=\frac{51578}{59739}\approx 86.34\%$. For more experiments, see \Cref{sec:compare}.

\subsection{Comparison with Existing Algorithms}\label{sec:compare}
After fixing the ordering, we can compare our algorithm with existing clustering algorithms combined with dimensionality reduction methods.
Explainable clustering algorithms all rely on decision trees, which take $d$ rounds in MPC on a $d$-dimensional point set as we proved in earlier sections. This is $42$ rounds for the Email dataset and $6$ rounds for the Cisco dataset. Also, there is an exponential dependency on $d$ which does not allow us to run those algorithms on datasets of tens of thousands of points with more than two dimensions.
DBSCAN (mathematically defined as the connected components of the disk graph) and $k$-medoids ($k$-means using the input points as centers) are used as subroutines in several subspace clustering algorithms.

We repeat the experiments of the previous sections for DBSCAN using radius $10$ and using $\ell_{\infty}$ distance after applying LSH and $k$-medoids using the number of clusters as $k$. Since the datasets are too big for these algorithms and larger clusters are better for comparing algorithms, we use subsets of the datasets. Still, for $k$-medoids we can only verify the labeling.

\Cref{table:experiment} shows the summary of the results of our experiments on big data sets. DBSCAN is computed using spatial indexing and blocks before reordering the input or using disjoint clusters. Heat map sorting is the algorithm of \Cref{sec:cisco} and LSH is the algorithm of \Cref{sec:dim}. The subset of the second dataset that we use is from a region where several clusters intersect and it covers most of the input points.

For the subset of the Email dataset used in \Cref{sec:email} with $n=5940$, we must randomly sample $9$ dimensions (source vertices) since $\ln n <9$, and repeat this $n^{1/5}<6$ times (for $6$ subspaces) which gives a $5$-approximation of the distances. The number of clusters in this subset is $7$.
Since this is the verification problem, we set both radius and $k$. The radius of $k$-medoids for $k=6$ is set to $100$.
For a subset of Cisco dataset, $n=57913$, $\ln n<10$, $n^{1/5}<9$.

\begin{table*}[h]
\centering
\begin{tabular}{|p{1.8cm}|c|c|c|c|}
\hline
Algorithm & \multicolumn{2}{|c|}{Email Dataset} & \multicolumn{2}{|c|}{Cisco Dataset}\\
& $\rho$ & $k\%$ & $\rho$ & $k\%$\\
\hline
\hline
DBSCAN & $0.97\%$ & $\frac{4}{42}\approx 9.52\%$ & $41.45\%^*$ & $\frac{5}{12}\approx 41.67\%^*$\\
\hline
LSH + $k$-medoids & $\frac{5736}{5940}\approx96.57\%^*$ & $100\%^*$  & $\frac{43328}{57913}\approx74.82\%^*$ & $\frac{4}{5}\approx 80\%^*$ \\ 
\hline
HMS & $81.83\%^*$ & $100\%^*$ & $100\%^*$ & $\frac{11}{12}\approx 91.67\%^*$\\
\hline
\end{tabular}
\caption{Comparison of heat map sorting algorithms on big data sets. Experiments that used a subset of the dataset were marked with $^*$.}\label{table:experiment}
\end{table*}

\Cref{table:experiment} compares the accuracy of the explainable model $(\rho)$ and the ratio of the clusters preserved by the algorithm to the input clusters $(k\%)$.
As expected, the heat map sorting algorithm outperforms existing algorithms, except for the case where the clusters are in the shape of a set of disjoint disks (the Email dataset has clusters in the shape of squares, which is close to circles) and there is a uniform noise. In this special case, the outliers are similar to each other and non-outlier points, and the clusters are convex, $k$-means algorithm easily removes the outliers by considering only convex clusters while DBSCAN fails to do so due to its sensitivity to the input radius parameter. For example, on the Email dataset, DBSCAN has less than 1\% accuracy. An example of a cluster that DBSCAN misses is \Cref{fig:cluster} where some extra rows/columns and diagonal connections (instead of the same row/column) split the cluster.
Decision trees are also sensitive to $k$ and running single-linkage to convert a bound on the distance (similarity) to the number of clusters requires at least a logarithmic number of rounds.
\section{Conclusions and Open Problems}
We discussed MPC algorithms for HMS and its special cases, based on our serial parameterized and approximation algorithms and special cases of parallel algorithms that we proved had better complexity than the general case.

Some of the advantages of heat map sorting are:
\begin{itemize}
\item a representation method for the relation between the dimensions and the points of different clusters simultaneously,
\item compressing the data on both the dimensions and the input points, and
\item a better visibility of a large dataset with non-numerical attributes.
\end{itemize}
Also, parallel HMS reduces graph problems to an embedding (a geometric problem) while preserving properties such as connectivity.

In computer network applications, using explainable clustering on time intervals, IP ranges, and protocols or port ranges can help identify recurring attacks or problems. This can be used as a tool in secondary data analysis or network traffic analysis by using one of the attributes as the label and the rest as dimensions.

We discussed the sensitivity of algorithms to their input parameters (the effect of small perturbations on the performance) of clustering algorithms. More specifically, we showed that the sensitivity of DBSCAN to its radius is more than the sensitivity of HMS to its radius, and the sensitivity of partitioning-based methods such as decision trees and HMS with disjoint rectangles to the number of clusters $(k)$ is more than $k$-means and HMS with dimension-disjoint rectangles.

Our proofs depend on the existence of a solution that preserves all the clusters but we did not prioritize between keeping the maximum number of clusters and minimizing the number of misclassifications (outliers). Prioritizing one might be useful in scenarios where almost all the points or almost all the clusters can be preserved.

\bibliographystyle{unsrt}
\bibliography{refs}
\newpage
\appendix
\section{Prerequisites}
\subsection{Approximation Algorithms}
Approximation algorithms are polynomial-time algorithms that compute a solution that is at most $\alpha$ time the optimal solution for minimization problems and at least $1/\alpha$ time the optimal solution for maximization problems, for a known function $\alpha(n)$ which is called the approximation ratio and $n$ is the size of the input.

\subsection{Massively Parallel Computation (MPC)}\label{sec:mpc}
A theoretical model for map-reduce considers a set of $L$ machines, each with memory $m$, that independently and in parallel process data during a set of $R$ rounds. We use the massively parallel computation (MPC) model~\cite{mpc,mpcj}, where $m=O(n^{\eta}), L=O(n/m), R=O(1)$, for some constant $\eta\in (0,1)$, however, polylogarithmic factors of $n$ are also allowed as multiples of $m$ and $L$ in most cases.

The computational complexities of MPC and a less restricting model that still contains MPC, called the map-reduce class (MRC)~\cite{mrc}, have been studied with relation to TISP class and the BSP model~\cite{fish2015computational}, as well as the NC class, circuit complexity, and PRAM models~\cite{frei2019efficient}. Checking whether two vertices $s$ and $t$ are connected (known as the st-connectivity problem) is hard to compute in $O(1)$ rounds of MPC~\cite{yaroslavtsev2018massively,andoni2018parallel}. Similar results exist for the hardness of st-connectivity in other parallel models.

Parallel algorithms for P-complete problems and logarithmic-time parallel algorithms (in the PRAM model) for connectivity in the NC class are open.

\subsubsection{Fixed-Parameter Algorithms}\label{sec:fixedparameter}
There are no known polynomial-time algorithms for NP-hard problems, however, fixed-parameter algorithms run in polynomial-time algorithms at the cost generality of the solution.
A problem is fixed-parameter tractable (FPT) if for a parameter $k$ of the input, the problem can be solved in time complexity that is an arbitrary function of $k$ and polynomial in $n$, where $n$ is the size of the input.

\end{document}